\documentclass[journal,10pt]{IEEEtran}
\usepackage[T1]{fontenc}
\usepackage{cite}
\usepackage{amsmath,amssymb,amsthm,amsfonts}
\usepackage{graphicx}
\usepackage{microtype}
\graphicspath{{figs/}}
\newtheorem{theorem}{Theorem}
\newtheorem{proposition}{Proposition}
\newtheorem{lemma}{Lemma}
\newtheorem{corollary}{Corollary}
\newtheorem{remark}{Remark}

\newcommand{\SO}{\mathrm{SO}(3)}
\newcommand{\Sph}{\mathbb{S}^2}
\newcommand{\C}{\mathbb{C}}
\newcommand{\R}{\mathbb{R}}

\title{Orientation Diversity Provably Increases the\\Angular-Domain Rank of 3D MIMO Arrays}

\author{Samer~Henry and Elvino~Sousa%
\thanks{The authors are with the Department of Electrical and Computer
Engineering, University of Toronto, Toronto, ON, Canada (e-mail:
sam.henry@mail.utoronto.ca; es.sousa@utoronto.ca).}%
\thanks{Accepted for publication in \emph{IEEE Wireless Communications
Letters}. \copyright~2026 IEEE. Personal use of this material is permitted.
Permission from IEEE must be obtained for all other uses, in any current or
future media, including reprinting/republishing this material for advertising
or promotional purposes, creating new collective works, for resale or
redistribution to servers or lists, or reuse of any copyrighted component of
this work in other works.}}

\begin{document}
\maketitle

\begin{abstract}
Movable- and rotatable-antenna systems treat element placement as a design
resource, but the supporting theory is largely position-only. We prove that
element \emph{orientation} is an effective spatial resource in its own right,
complementary to position. For collinear or coplanar arrays, the angular-domain rank is capped
by the number of projection-ambiguity classes of the scattering directions,
for \emph{every} common element orientation and pattern. Orientation-diverse
elements with sufficiently rich analytic patterns restore the maximal rank
$\min(N,L)$ for almost every orientation draw. A closed-form correlation
expression, strictly monotone in orientation spread under isotropic
scattering, quantifies the gain within the scalar-pattern model, and Monte
Carlo simulations corroborate the exact-rank and capacity behavior.
\end{abstract}

\begin{IEEEkeywords}
Angular-domain rank, antenna arrays, element orientation, MIMO, rotatable
antennas, spatial correlation.
\end{IEEEkeywords}

% ============================================================================
\section{Introduction}
\label{sec:intro}

\IEEEPARstart{F}{lexible}-antenna architectures have made element placement a
first-class optimization variable: movable antennas reposition elements
continuously \cite{ZhuMag24,ZhuTWC24}, fluid antennas select among ports
\cite{Wong21}, and six-dimensional movable antenna (6DMA) systems additionally
\emph{rotate} them \cite{ShaoTWC25,ShaoJSAC25,ZhengTCOM26}; related flexible
architectures extend to non-terrestrial and reconfigurable-intelligent-surface-aided
networks \cite{Lin21,Guo22}. The supporting
theory, however, is almost entirely position-only: capacity and
degrees-of-freedom (DoF) analyses model reconfiguration through the phase
profile $e^{\jmath\kappa p^{\mathsf T}u}$ of element positions $p$, inheriting
the classical angular-domain framework of uniform linear arrays (ULAs)
\cite{TseViswanath05,Sayeed02,PoonBrodersenTse05}. The 6DMA literature models rotation analytically through directive gain
and optimizes rotation numerically \cite{ShaoTWC25}; what is missing is a rank-theoretic
converse--achievability pair for rotation. Element \emph{response}
diversity is known to carry spatial DoF: polarimetric elements provably gain
DoF \cite{AndrewsMitra01,PoonTse11}, and pattern diversity decorrelates
co-located branches empirically \cite{vaughan1987,dietrich2001}. In the
vector model, rotating an element is a special case of response diversity;
this letter isolates the scalar-pattern mechanism and, at arbitrary fixed
positions, supplies finite-$N$ exact-rank statements with a converse,
complementing the asymptotic DoF results of \cite{PoonTse11}.
Non-uniform position geometries are likewise known to enhance identifiable
DoF \cite{PalVaidyanathan10}; empirically, three-dimensional (3D) arrays
with non-uniform element orientations show higher angular rank than planar
arrays of equal size \cite{Henry18,Henry25}.

This letter closes the gap with three exact results. \emph{(i)} A converse:
for arrays whose positions lie on a line or plane, the angular-domain rank is
capped by the number of \emph{projection-ambiguity classes} of the scattering
directions (cones for linear arrays, mirror pairs for planar arrays), for
every common element orientation and every pattern
(Proposition~\ref{prop:ceiling}). \emph{(ii)} An achievability result:
orientation-diverse elements with a sufficiently rich analytic pattern break
these ambiguities and attain the maximal rank $\min(N,L)$ for almost every
orientation draw, at \emph{arbitrary} fixed positions
(Theorem~\ref{thm:main}). \emph{(iii)} A closed form: for two coincident cosine-power elements, the
correlation is derived exactly as a strictly monotone function of
orientation spread (Proposition~\ref{prop:closed}), turning orientation
spread into a quantified design rule. Simulations confirm the theory.

% ============================================================================
\section{System Model}
\label{sec:model}

$N$ receive elements occupy positions $p_i\in\R^3$ with orientations
$R_i\in\SO$; all share a complex voltage pattern $f:\Sph\to\C$ evaluated in
the local frame, so a unit plane wave from direction $u\in\Sph$ produces
\begin{equation}
  [a(u)]_i \;=\; f(R_i^{\mathsf T}u)\,e^{\jmath\kappa p_i^{\mathsf T}u},
  \qquad i=1,\dots,N,
  \label{eq:steer}
\end{equation}
with wavenumber $\kappa=2\pi/\lambda_c$ and carrier wavelength $\lambda_c$. The narrowband channel over $L$
paths with distinct directions $u_1,\dots,u_L\in\Sph$ and nonzero gains $\beta_l$ is
$h=\sum_{l}\beta_l\,a(u_l)$; with a scattering measure $\mu$ on $\Sph$, the
receive covariance is $R_a=\int a(u)a(u)^{\mathsf H}\,d\mu(u)$. The
\emph{angular rank} over a direction set $\Theta$ is
\begin{equation}
  \zeta(\Theta)\;\triangleq\;\dim\,\mathrm{span}\{a(u):u\in\Theta\},
  \label{eq:zeta}
\end{equation}
which for $L$ discrete paths equals $\operatorname{rank}A$,
$A=[a(u_1)\cdots a(u_L)]$, and satisfies
$\operatorname{rank}R_a=\zeta(\operatorname{supp}\mu)$ for continuous $f$
(since $x^{\mathsf H}R_ax=\int|a(u)^{\mathsf H}x|^2d\mu$). For a
half-wavelength ULA of isotropic elements, \eqref{eq:zeta} coincides with
the classical count of occupied virtual (beamspace) bins
\cite{Sayeed02,TseViswanath05} when each bin holds at most one path, and
refines it otherwise; $\zeta$ upper-bounds the multiple-input
multiple-output (MIMO) channel rank, with equality when the transmit side
resolves the paths with nonzero gains. Fig.~\ref{fig:amb}
illustrates the setting. We assume $f$ real-analytic and not identically
zero; the pattern-richness condition below is stated via the spherical
harmonic degree support $\mathcal N(f)=\{n:\exists m,\
\langle f,Y_n^m\rangle\neq0\}$.

\begin{figure}[t]
  \centering
  \includegraphics[width=0.96\columnwidth]{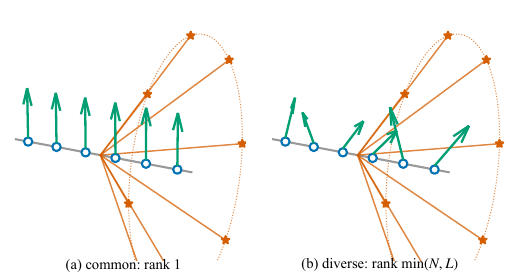}
  \caption{Collinear array receiving $L$ paths on one cone of directions
    (equal directional cosines). (a) With a common element orientation the
    steering vectors coincide up to scale: rank $1$. (b) Orientation-diverse
    directional elements separate the paths: rank $\min(N,L)$ for almost
    every orientation draw (Theorem~\ref{thm:main}).}
  \label{fig:amb}
\end{figure}

% ============================================================================
\section{Main Results}
\label{sec:results}

\subsection{The Positional Ceiling}

Let the positions lie in an affine set $V=q_0+V_0$ with
$\dim V_0=k\le2$, and call $u,u'$ \emph{$V$-equivalent} if
$\Pi_{V_0}(u-u')=0$, where $\Pi_{V_0}$ projects onto $V_0$. For $k=1$
(collinear array along $d$) the classes are the cones
$\{u: d^{\mathsf T}u=\Omega\}$, $\Omega\in[-1,1]$, of equal directional
cosine; for $k=2$
(coplanar) they are mirror pairs about the array plane. Let $L'$ be the
number of classes among $u_1,\dots,u_L$.

\begin{proposition}[Positional ceiling]
\label{prop:ceiling}
If all elements share one orientation, $R_i=\bar R$, then for every
$\bar R\in\SO$ and every pattern $f$,
$\operatorname{rank}A\le\min(N,L')$. In particular, $L$ paths on a single
cone yield $\operatorname{rank}A\le1$ for every collinear placement and
every common orientation.
\end{proposition}
\begin{IEEEproof}
With $p_i=q_0+v_i$, the $i$-dependence of $[a(u_l)]_i$ enters only
through $e^{\jmath\kappa(\Pi_{V_0}u_l)^{\mathsf T}v_i}$, one fixed vector per
class; each column of $A$ is a scalar multiple of its class vector.
\end{IEEEproof}

The cone and mirror ambiguities themselves are classical
\cite{VanTrees02}; Proposition~\ref{prop:ceiling} adds the quantifier that
matters here: \emph{no} common rotation and \emph{no} pattern can remove
them, making rigorous the projection heuristic of \cite{Henry18}. Conversely, for
\emph{generic} positions filling $\R^3$, standard character-independence
arguments give $\operatorname{rank}A=\min(N,L)$ almost everywhere; the
interesting regime is constrained form factors, where positions are not free.

\subsection{Orientation Diversity Restores Full Rank}

The engine of the main result is a complete characterization of when rotated
copies of a pattern are linearly independent. For distinct
$u_1,\dots,u_L$ define $G_n=[\,Y_n^{m'}(u_l)\,]_{m',l}\in\C^{(2n+1)\times L}$
and
\begin{equation}
  K_f \;=\; \textstyle\bigcap_{n\in\mathcal N(f)}\ker G_n \;\subseteq\;\C^L.
  \label{eq:Kf}
\end{equation}

Intuitively, $K_f$ is nontrivial exactly when the pattern is too smooth to
separate the $L$ directions under rotation.
\begin{lemma}[Rotated-pattern independence]
\label{lem:pw}
The space of linear relations among the functions
$h_l(R)=f(R^{\mathsf T}u_l)$ on $\SO$ is exactly $K_f$. In particular
$\{h_l\}$ is linearly independent if and only if $K_f=\{0\}$, which holds whenever
$\{0,1,\dots,L-1\}\subseteq\mathcal N(f)$, and for every $L$ whenever
$\mathcal N(f)$ contains infinitely many even and infinitely many odd
degrees \cite{ChenMenegattoSun03}.
\end{lemma}

The cosine-power pattern
$f(v)=\big(\tfrac{1+v^{\mathsf T}n_0}{2}\big)^{q}=\cos^{2q}(\psi/2)$
(boresight $n_0$, boresight angle $\psi$) has
$\mathcal N(f)=\{0,\dots,q\}$, hence resolves any $L\le q+1$ distinct
directions. Patterns with front--back symmetry, $f(-v)=f(v)$, can never
separate antipodal paths; richness is a substantive requirement rather than a technicality.

\begin{theorem}[Orientation diversity]
\label{thm:main}
Let $u_1,\dots,u_L$ be distinct and let the positions $p_1,\dots,p_N$ be
arbitrary and fixed (possibly collinear, coplanar, or coincident).
(a)~If the analytic pattern $f$ satisfies $K_f=\{0\}$, then for
Haar-almost every $(R_1,\dots,R_N)\in\SO^N$,
\begin{equation}
  \operatorname{rank}A=\min(N,L).
  \label{eq:main}
\end{equation}
(b)~If instead $f$ is any analytic pattern and the positional phases factor
as $e^{\jmath\kappa p_i^{\mathsf T}u_l}=\alpha_i b_l$ (coincident
elements, or a collinear array with cone-supported scattering), then for
Haar-almost every draw $\operatorname{rank}A=\min(N,\,L-\dim K_f)$,
interpolating between the isotropic collapse ($\dim K_f=L-1$, rank $1$)
and full rank ($K_f=\{0\}$).
\end{theorem}

Together, Proposition~\ref{prop:ceiling} and Theorem~\ref{thm:main} quantify
the \emph{angular-domain gain of orientation}: the rank jumps from
$\min(N,L')$ to $\min(N,L)$, and the gap can be as large as $N-1$. The
proof of Theorem~\ref{thm:main} (Appendix~\ref{app:main}) combines
Lemma~\ref{lem:pw}, a determinant induction with row-wise nonzero weights,
and the fact that nontrivial real-analytic functions on $\SO^N$ vanish only
on Haar-null sets.
\begin{remark}
\label{rem:scope}
Three qualifications delimit the claim. First, the rank gap requires
(near-)degenerate scattering geometry ($L'<L$): for generic directions
$L'=L$, and a common orientation already suffices, so away from degeneracy
the benefit of orientation diversity manifests as \emph{conditioning},
quantified here only for $N=2$ (Proposition~\ref{prop:closed}). Second, the
almost-sure statement certifies no specific design; a chosen orientation
tuple is validated by verifying that the minor $\Psi$ of
Appendix~\ref{app:main} is nonzero. Third, \emph{guaranteeing} resolution of
any $L$ directions uses harmonic content up to degree $L-1$ (Lemma~\ref{lem:pw});
paths within one ambiguity class are concyclic, forcing degree at least
$\lceil(L-1)/2\rceil$, i.e., increasingly directive elements. The exact
per-class rank is $\min(L,2q+1)$, of which roughly $q+1$ dimensions remain well
conditioned at practical dynamic range ($2$--$3$ for patch-like $q\le2$).
\end{remark}

\begin{corollary}[Covariance version]
\label{cor:cov}
Let $\mu$ be a scattering measure whose support is contained in $K$ cones
of a collinear array. For every common orientation,
$\operatorname{rank}R_a\le K$. If instead the orientations are drawn from Haar measure, the positions are
arbitrary and fixed, $f$ is analytic with
$\mathcal N(f)\supseteq\{0,\dots,N-1\}$ (or with infinitely many even and
infinitely many odd degrees in $\mathcal N(f)$), and
$\operatorname{supp}\mu$ contains at least $N$ distinct directions, then
$\operatorname{rank}R_a=N$ and $\lambda_{\min}(R_a)>0$ for Haar-almost every
draw: the rank of the receive correlation jumps from at most $K$ to $N$.
\end{corollary}

Corollary~\ref{cor:cov} is the measure-theoretic form of the ceiling: the
discrete class count $L'$ is replaced by the number of projection-ambiguity
classes met by $\operatorname{supp}\mu$, so a continuous scattering measure
concentrated on $K$ cones is capped at $\operatorname{rank}R_a\le K$ for every
common orientation. The achievability half is
unchanged, since $\operatorname{rank}R_a=\zeta(\operatorname{supp}\mu)$ by
\eqref{eq:zeta}.

\subsection{A Closed Form, Monotone in Orientation Spread}

\begin{proposition}[Two-element closed form]
\label{prop:closed}
Let two coincident elements have pattern $\cos^{2q}(\psi/2)$ with boresights
separated by $\Delta\in[0,\pi]$, under isotropic scattering ($\mu$ the
uniform probability measure on $\Sph$). Then $R_a$ has diagonal
$\gamma=\tfrac{1}{2q+1}$ and off-diagonal $\rho(\Delta)\ge0$,
\begin{equation}
  \rho(\Delta)=\frac{(q!)^2}{(2q+1)!}\sum_{k=0}^{q}\binom{q}{k}
  \binom{q+k}{k}\,c^{k}(1-c)^{q-k},
  \label{eq:rho}
\end{equation}
where $c=\cos^{2}(\Delta/2)$. $\rho$ is strictly decreasing on $[0,\pi]$;
hence $\lambda_{\min}=\gamma-\rho(\Delta)$ strictly increases and the
condition number $\xi(\Delta)=\frac{\gamma+\rho}{\gamma-\rho}$ strictly
decreases on $(0,\pi]$. Endpoints: $\rho(0)=\gamma$ (rank $1$,
$\xi\to\infty$) and $\rho(\pi)/\gamma=1/\binom{2q}{q}$, so
$\xi(\pi)=\big(\binom{2q}{q}+1\big)/\big(\binom{2q}{q}-1\big)$, e.g.,
$\xi(\pi)=3$, $1.4$, $21/19$ for $q=1,2,3$.
\end{proposition}

Proposition~\ref{prop:closed} (proof in Appendix~\ref{app:closed}) is the
orientation-domain counterpart of the classical distance--correlation law:
correlation decreases monotonically in orientation separation at zero
positional separation, with an exact expression. It supplies the quantitative
design rule visible in Fig.~\ref{fig:corr}: at $\Delta=\pi$, two $q=3$ elements are nearly uncorrelated (normalized
correlation $r(\Delta)=\rho/\gamma=0.05$) despite sharing a position.

\begin{figure}[t]
  \centering
  \includegraphics[width=0.94\columnwidth]{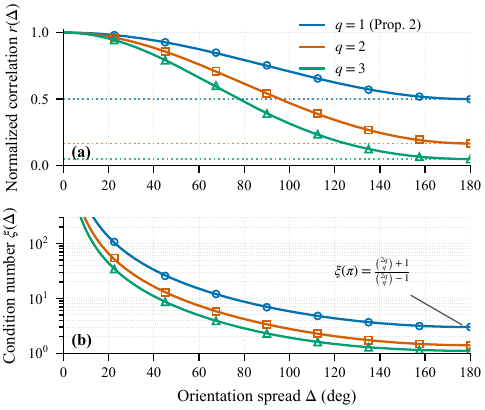}
  \caption{Proposition~\ref{prop:closed}: (a) normalized correlation
    $r(\Delta)=\rho/\gamma$ and (b) condition number $\xi(\Delta)$ for
    $q\in\{1,2,3\}$. Solid lines: exact closed form \eqref{eq:rho}; markers:
    Monte Carlo over $10^3$ isotropic-scattering realizations; dotted lines:
    exact endpoints $1/\binom{2q}{q}$.}
  \label{fig:corr}
\end{figure}

% ============================================================================
\section{Numerical Validation}
\label{sec:num}

Fig.~\ref{fig:corr} overlays Monte Carlo estimates of $r(\Delta)$ and
$\xi(\Delta)$ on the exact expressions; standard errors are below marker
size at $10^3$ realizations.

Fig.~\ref{fig:rank} instantiates Proposition~\ref{prop:ceiling},
Theorem~\ref{thm:main}, and Corollary~\ref{cor:cov} in a deliberately
degenerate worst-case geometry: $L=24$ unit-gain paths placed in
mirror-symmetric pairs on $K=3$ cones of a half-wavelength collinear array
($K'=12$ mirror classes for the coplanar comparison). Orientations are drawn
from Haar measure with pattern $f(v)=e^{v^{\mathsf T}n_0}$ (both parities,
so $K_f=\{0\}$ by Lemma~\ref{lem:pw}). Panel~(a) reports the numerical rank of $R_a$ at
machine tolerance: arrays with a common orientation saturate at their ambiguity
counts ($3$ and $12$) for every draw, while orientation-diverse elements
attain $\min(N,L)$, an eightfold exact-rank gain at $N=24$ in this
geometry. Consistent with Remark~\ref{rem:scope}, the thresholded eigenvalue
count at finite dynamic range is smaller; panel~(b) therefore reports the
operational quantity, the capacity
$C=\log_2\det\!\big(I+\tfrac{\rho}{N}R_a\big)$ at $\rho=30$\,dB (the
DoF-limited regime) under the equal-trace normalization
$\operatorname{tr}R_a=N$, averaged over orientation draws: orientation
diversity lifts the collinear array from its saturated $25.0$\,b/s/Hz to
$44.4$\,b/s/Hz at $N=24$, recovering most of the planar aperture's advantage
without leaving the line.

\begin{figure}[t]
  \centering
  \includegraphics[width=0.94\columnwidth]{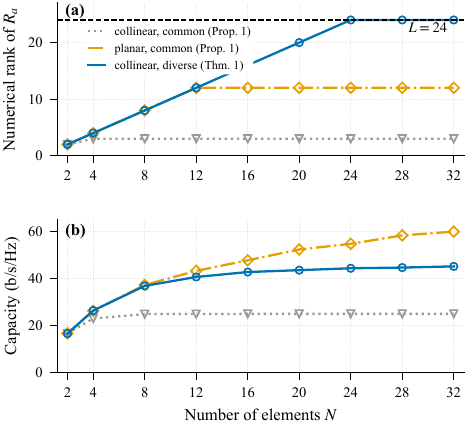}
  \caption{(a)~Numerical rank of $R_a$ at machine tolerance and
    (b)~capacity $\log_2\det(I+\tfrac{\rho}{N}R_a)$, $\rho=30$\,dB,
    $\operatorname{tr}R_a=N$, versus
    $N$ ($L=24$ paths in mirror pairs on $K=3$ cones; $K'=12$ mirror
    classes; pattern $e^{v^{\mathsf T}n_0}$). (a)~lines: exact ranks
    (Prop.~\ref{prop:ceiling}, Thm.~\ref{thm:main}), markers: Monte Carlo;
    (b)~Monte Carlo means ($10^3$ draws).}
  \label{fig:rank}
\end{figure}

\begin{figure*}[t]
  \centering
  \includegraphics[width=0.98\textwidth]{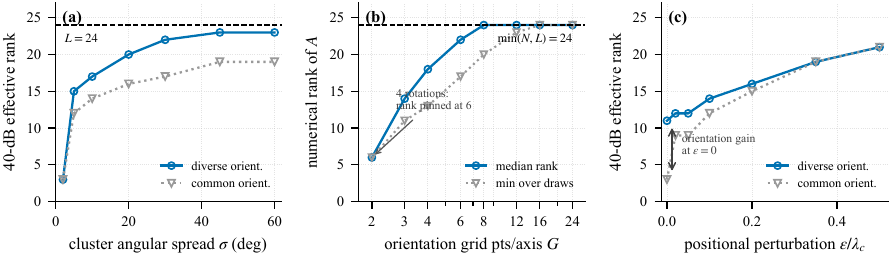}
  \caption{Robustness beyond the worst-case geometry of Fig.~\ref{fig:rank}
    ($N=L=24$, collinear array, pattern $e^{v^{\mathsf T}n_0}$).
    (a)~Von Mises--Fisher clustered scattering, $8$ paths per cluster on the
    $3$ cone centers, spread $\sigma$: median $40$-dB effective rank ($200$
    draws). (b)~Orientations drawn from a $G$-point-per-axis Euler grid:
    median and minimum numerical rank ($400$ draws). (c)~Positions perturbed
    off the array line by $\epsilon/\lambda_c$: median $40$-dB effective rank
    ($300$ draws).}
  \label{fig:robust}
\end{figure*}

Fig.~\ref{fig:robust} stress-tests the mechanism along three practical axes;
panels (a) and (c) report the $40$-dB \emph{effective} rank of $R_a$
(eigenvalues within $40$\,dB of the largest). \emph{(a)~Anisotropic scattering.} Relaxing
the cones to von Mises--Fisher clusters of spread $\sigma$ (mmWave-like sparse
multipath) makes the degeneracy generic ($L'=L$), and both arrays resolve more
modes as $\sigma$ grows; orientation diversity resolves $3$--$5$ more than a
common orientation across the sweep: away from exact degeneracy the gain takes
the conditioning form anticipated by Remark~\ref{rem:scope} (as $\sigma\to0$
the clusters merge into three sources no array resolves).
Prop.~\ref{prop:closed}'s monotone law is an isotropic-ensemble statement: for
one fixed cluster $\rho(\Delta)$ peaks as the second boresight sweeps past it;
averaged over cluster placements the monotone decrease is restored (not
shown). \emph{(b)~Discrete orientation grids.}
On a grid, rank fails only on the null variety $\{\Psi=0\}$ of the certifying
minor; a tuple is certified by checking $\Psi\neq0$ (Remark~\ref{rem:scope}).
Panel~(b) exhibits the threshold: the
four-rotation $G=2$ grid pins the rank at $6$, the median reaches $\min(N,L)$
by $G=8$, and from $G=16$ every sampled tuple is full rank: past the
threshold, quantization is immaterial to rank. \emph{(c)~Near-collinear
positions.} Theorem~\ref{thm:main}(b) ties position degeneracy to recoverable
rank via $L-\dim K_f$. At $\epsilon=0$ a common orientation resolves only the
$K=3$ cone classes while orientation diversity already resolves $11$ modes;
positional perturbation alone matches this only near
$\epsilon\approx0.35\lambda_c$: orientation supplies at $\epsilon=0$ what
position buys only by leaving the form factor.

% ============================================================================
\section{Conclusion}
\label{sec:conc}

Within a scalar far-field model, element orientation is a provably
effective spatial resource, complementary to position: it breaks the
projection ambiguities that cap the angular rank of collinear and coplanar
arrays, restores $\min(N,L)$ almost surely under an explicit
pattern-richness condition, and admits an exact monotone correlation law in
the orientation spread. The model omits polarization and mutual coupling;
pattern richness of degree $q$ implies element electrical size
$\kappa a\gtrsim q$ for element radius $a$. Polarization is the key omission: rotating a physical element also rotates its
polarization ellipse, so the rotation studied here is one component of a
full $\SO$ action on a dual-polarized pattern. We conjecture that
the two effects \emph{compound} rather than compete: orientation breaks the
scalar projection ambiguities, while polarization contributes up
to a further factor of two of rank per resolved direction in the far field
\cite{PoonTse11} (cf.\ the tripling of \cite{AndrewsMitra01}); the factor of
two is conjectural, contingent on independent polarimetric responses, not a
general result. The converse
(Proposition~\ref{prop:ceiling}) holds verbatim for single-port elements with a
vector effective length; for dual-polarized elements each ambiguity class
carries up to two polarization dimensions, and the same argument gives the
relaxed ceiling $\min(N,2L')$, quantitatively consistent with the conjecture.
Only the achievability side (Theorem~\ref{thm:main}) requires a polarimetric
generalization; this is left open. These results supply the missing
rank-theoretic foundation for rotatable and 6DMA architectures
\cite{ShaoTWC25,ZhengTCOM26} and for 3D conformal arrays on constrained
platforms, where positions cannot be freely chosen but orientations can.

% ============================================================================
\appendices
\section{Proof of Theorem~\ref{thm:main}}
\label{app:main}

\emph{Lemma~\ref{lem:pw}:} Expanding $f=\sum_{n,m}a_{nm}Y_n^m$ and using
the rotation law $Y_n^m(R^{-1}u)=\sum_{m'}D^n_{m'm}(R)Y_n^{m'}(u)$ (which
fixes our Wigner convention) with Schur orthogonality of the $D^n_{m'm}$
yields $\langle\sum_l\lambda_l h_l,\,D^n_{m'm}\rangle
=\tfrac{a_{nm}}{2n+1}[G_n\lambda]_{m'}$; completeness (Peter--Weyl) makes
the relation space exactly $K_f$ in \eqref{eq:Kf}. Sufficiency of degrees
$\{0,\dots,L-1\}$: the polynomial $\prod_{m\neq l}(1-u_m^{\mathsf T}v)$
has degree $L-1$ and vanishes at $u_m$ for every $m\neq l$ but not at $u_l$. For the
both-parity case, form the kernel $\sum_{n\in\mathcal N(f)}2^{-n}
\sum_m Y_n^m(u)\overline{Y_n^m(v)}$; a relation $\lambda\in K_f$
annihilates its quadratic form, while \cite{ChenMenegattoSun03} (applied to
$\mathrm{Re}\,\lambda$ and $\mathrm{Im}\,\lambda$) forces $\lambda=0$.

\emph{Theorem:} Let $c_{il}=e^{\jmath\kappa p_i^{\mathsf T}u_l}\neq0$ and
$L\le N$ (else pass to an $N$-subfamily). An induction on $L$ shows that for
linearly independent $h_1,\dots,h_L$ and any nonzero constants $c_{il}$,
the function $\Psi(R_1,\dots,R_L)=\det[c_{il}h_l(R_i)]$ is not identically
zero: upon expansion along the first row, the coefficient of $h_L(R_1)$ is a
nonzero multiple of the inductively nonvanishing minor, and a nontrivial
combination of independent functions is not the zero function. $|\Psi|^2$ is
real-analytic on the compact connected group $\SO^L$, and a nontrivial
real-analytic function on a second-countable connected analytic manifold
vanishes on a null set \cite{KrantzParks02}; the product with
$\SO^{N-L}$ remains null. Hence the top
$L\times L$ minor of $A$ is nonzero almost surely. For the factored-phase
case, the relations $\lambda\in K_f$ transfer to column relations of $A$
valid for every orientation tuple (with $f$ continuous), giving the upper
bound $L-\dim K_f$; achievability repeats the minor argument on a maximal
independent subfamily of $\{h_l\}$, whose size is $L-\dim K_f$ by
Lemma~\ref{lem:pw}. \hfill$\blacksquare$

\section{Proof of Proposition~\ref{prop:closed}}
\label{app:closed}

With polar angle $\psi$ from the first boresight $n_0$, azimuth $\phi$ with
$n_1$ in the $\phi=0$ half-plane, $\mu$ the uniform probability measure, and
$s=\cos^2(\psi/2)$ ($\mu$ uniform in $s$ on $[0,1]$), the diagonal is
$\gamma=\int_0^1 s^{2q}\,ds=\tfrac1{2q+1}$. For boresights separated by
$\Delta$, the identity
$\tfrac{1+u^{\mathsf T}n_1}{2}=\lvert\alpha+\beta e^{-\jmath\phi}\rvert^2$
with $\alpha=\cos\tfrac\Delta2\cos\tfrac\psi2$,
$\beta=\sin\tfrac\Delta2\sin\tfrac\psi2$ reduces the sphere average to a
binomial $\phi$-average and a Beta integral:
$\rho=\sum_k\binom qk^2c^k(1-c)^{q-k}\tfrac{(q+k)!(q-k)!}{(2q+1)!}$
(re-indexing $k\leftrightarrow q-k$), which equals \eqref{eq:rho} by
$\binom qk^2(q+k)!(q-k)!=(q!)^2\binom qk\binom{q+k}k$. Monotonicity follows from $\frac{d}{dc}\sum_k g_k b_{k,q}(c)
=q\sum_k(g_{k+1}-g_k)b_{k,q-1}(c)>0$ for the Bernstein basis $b_{k,q}$ with
$g_k=\binom{q+k}{k}$ strictly increasing; endpoints follow by substitution.
\hfill$\blacksquare$

% ============================================================================
\IEEEtriggeratref{11}

\end{document}